\documentclass[journal]{IEEEtran}

\usepackage[T1]{fontenc}
\usepackage[utf8]{inputenc}
\usepackage{amsmath,amssymb,amsfonts,amsthm}
\usepackage{bm,cite,graphicx,booktabs,microtype,url,xcolor}
\graphicspath{{figures/}}
\newtheorem{theorem}{Theorem}
\newtheorem{lemma}{Lemma}
\newtheorem{corollary}{Corollary}
\newtheorem{proposition}{Proposition}
\DeclareMathOperator{\sgn}{sgn}
\begin{document}

\title{QoS-Aware Joint 1-Bit RIS Element Activation and\\ Fluid-Antenna Port Selection}

\author{Khalid~Yahya,~\IEEEmembership{Member,~IEEE}, 
Mahmoud~Aldababsa,~\IEEEmembership{Senior Member,~IEEE}%%%
}

\markboth{Submitted to IEEE Journal,~VOL.~XX, NO.~XX,~2026}%
{Aldababsa: QoS-Aware Joint 1-Bit RIS Element Activation and Fluid-Antenna Port Selection}
\maketitle
{\color{black}
{\color{black}
\begin{abstract}
This letter jointly optimizes fluid-antenna port selection, fixed-cardinality reconfigurable intelligent surface (RIS) element activation, and 1-bit reflection phases through an exact angular event sweep that extends full-activation divide-and-sort geometry to subset selection. We establish deterministic gain bounds and activation monotonicity, derive exact single-element outage expressions and high signal-to-noise ratio (SNR) asymptotics for independent RIS branches with correlated receiving ports, and prove diversity order equal to the RIS element count for fixed positive-definite RIS covariance. With perfect instantaneous channel state information, a quality-of-service aware policy activates the minimum number of elements required to meet a target SNR while preserving full-activation outage probability. An explicit overhead condition identifies when reduced activation improves net energy efficiency. Correlated-channel benchmarks illustrate subset-selection gains, aperture effects, and overhead-dependent energy-efficiency tradeoffs.
\end{abstract}
}

{\color{black}
\section{Introduction}

Reconfigurable intelligent surfaces (RISs) enhance wireless
communication by controlling reflected propagation paths
\cite{basar2019wireless,huang2019ee}. Their performance and
energy consumption depend on both the reflection coefficients
and the number of active elements. Zappone et al.
\cite{zappone2021optimal} optimize the element count and
continuous phases while accounting for channel-estimation and
feedback overhead. With 1-bit phase control, each active element
has only two reflection phases, which couples subset selection
to coherent combining. Sign alignment (SA) provides simple
binary-phase configurations \cite{stylianopoulos2024sign},
whereas divide-and-sort (DaS) globally solves the full-activation
discrete inner-product problem and admits a simplified linear
search \cite{xiong2024das}. Neither method directly optimizes
a fixed-cardinality active subset.

A fluid antenna system (FAS) selects a receiving port within
a finite spatial aperture using a single radio-frequency (RF)
chain \cite{wong2021fas,bepari2026survey}. Closely spaced ports
are spatially correlated; consequently, increasing their number
does not create independent fading branches. Partial-observation
port selection \cite{chai2022port} and outage and
block-correlation analyses \cite{new2024outage,ramirez2024block}
address acquisition requirements and spatial dependence.
Zhao and Slock \cite{zhao2025insights} establish
covariance-rank-based diversity for point-to-point FAS.
In an RIS-assisted link, the transmitter--RIS channel associated
with each element is shared across receiving ports, requiring
a distinct cascaded-channel analysis. Moreover, outdated
channel state information (CSI) can separate the selected
configuration from its realized performance
\cite{psomas2023outdated}.
Related studies consider joint fluid-antenna (FA) and RIS
design under broader communication objectives. Zhang et al.
\cite{zhang2024idet} jointly optimize FA port selection and
transmit beamforming for multiuser integrated data and energy
transfer with imperfect CSI, without an RIS. Aldababsa et al.
\cite{aldababsa2026multiris} investigate joint RIS/port selection
and RIS selection with port combining in multi-RIS-assisted
FAS. Zhang et al. \cite{zhang2026fasris} develop a two-timescale
design using statistical CSI for FA/RIS configuration and
instantaneous CSI for precoding. Chen and Lee
\cite{chen2026unified} jointly optimize discrete FA positions,
fixed-cardinality irregular RIS or simultaneously transmitting
and reflecting RIS topology, and low-resolution
coefficients through cross-entropy probability learning.
These studies establish relevant activation and discrete-design
precedents. Unlike the probabilistic multiuser design in
\cite{chen2026unified}, this work derives a globally optimal
solution for the scalar fixed-cardinality binary-phase problem
with receiver-port selection and establishes deterministic
and outage guarantees.

The motivation is to characterize exactly how binary phase
control and cardinality-constrained activation interact with
receiving-port selection. Selecting the strongest cascaded
coefficients need not maximize the combined channel power,
because their phase compatibility also matters. Beyond
full-activation DaS, subset membership changes at additional
projection-ordering events. We derive an exact angular event
sweep for joint port selection, fixed-cardinality activation,
and 1-bit phase control, and establish deterministic gain bounds
and reliability monotonicity. We obtain exact single-element
outage expressions and high signal-to-noise ratio (SNR)
asymptotics for independent RIS branches with correlated FAS
ports, and extend the diversity-order result to fixed
positive-definite RIS covariance. Finally, we develop a
quality-of-service (QoS)-aware minimum-element policy under
perfect instantaneous CSI, prove its full-activation outage
equivalence and average-activation identity, and derive the
overhead condition required for a net energy-efficiency
improvement.

{\color{black}
%\noindent\textit{Notation:} $\mathcal{CN}(\bm\mu,\bm\Sigma)$ denotes a circular complex Gaussian law with mean $\bm\mu$ and covariance $\bm\Sigma$; $\bm0_k,\bm I_k$ are the zero vector and identity matrix; $\operatorname{vec}$ stacks columns. $\|\cdot\|_0$ counts nonzeros; $\|\cdot\|$ is Euclidean norm; $j^2=-1$. $\Re,\arg,\sgn$ denote real part, phase and sign; superscripts $T,H,*$ denote transpose, conjugate transpose and conjugation; $\otimes$ denotes Kronecker product.
}

{\color{black}
\section{System Model and Joint Selection}
A single-antenna transmitter serves an $M$-port single-RF-chain FAS through an $N$-element passive RIS. Let $g_n$, $H_{m,n}$ and $h_{d,m}$ denote transmitter--RIS, RIS--port, and direct channels, collected in $\bm g\in\mathbb C^N$ and $\bm H\in\mathbb C^{M\times N}$. Set $\mathcal N=\{1,\ldots,N\}$ and $\mathcal M=\{1,\ldots,M\}$.
The spatially correlated Rayleigh-fading model is $\bm g \sim\mathcal{CN}\!\left(\bm0_N,\Omega_g\bm R_R\right)$, and $\operatorname{vec}(\bm H)
 \sim\mathcal{CN}\!\left(
 \bm0_{MN},\Omega_h\bm R_R^T\!\otimes\bm R_F\right)$. The average powers are $\Omega_g,\Omega_h$; $\bm R_R,\bm R_F$ are unit-diagonal positive-semidefinite RIS and FAS covariance matrices. The hops are independent; the direct link is arbitrary for optimization and blocked for outage.
}
For $M>1$, the ports are uniformly spaced over a physical aperture of length $W\lambda$, where $W$ is the aperture normalized by the carrier wavelength $\lambda$. Under the adopted two-dimensional isotropic scattering model, the FAS correlation is $[\bm R_F]_{i,j}=J_0\!\left(2\pi W\frac{|i-j|}{M-1}\right), i,j\in\mathcal M$, where $J_0(\cdot)$ is the zeroth-order Bessel function of the first kind; for $M=1$, $\bm R_F=[1]$.
{\color{black}
An element is inactive ($u_n=0$) or reflects with unit magnitude and phase $0$ or $\pi$ ($u_n=\pm1$): $u_n\in\{-1,0,+1\}, \|\bm u\|_0=L, L\in\{1,\ldots,N\}$. Inactive elements have negligible residual reflection, requiring a matched or strongly attenuating state.
}
{\color{black}
Set $c_{m,n}=H_{m,n}g_n$, $\bm c_m=[c_{m,1},\ldots,c_{m,N}]^T$, and $\bar\gamma=P/\sigma^2$, where $P$ and $\sigma^2$ are transmit power and receiver-noise variance. The received signal at port $m$ is
\begin{equation}
 y_m=\sqrt{P}\left(h_{d,m}+\bm c_m^T\bm u\right)s+w_m,
\end{equation}
where $\mathbb E[|s|^2]=1$ and $w_m\sim\mathcal{CN}(0,\sigma^2)$ is independent of the signal and channels. The port SNR is
\begin{equation}
 \gamma_m(\bm u)=\bar\gamma\left|h_{d,m}+\bm c_m^T\bm u\right|^2.
 \label{eq:instantaneous_snr}
\end{equation}
The joint design is
\begin{equation}
 G^\star(L)=\max_{\substack{1\le m\le M,\;\bm u\in\{-1,0,1\}^N\\\|\bm u\|_0=L}}
 \left|h_{d,m}+\bm c_m^T\bm u\right|^2,
\label{eq:problem}
\end{equation}
The optimized SNR is $\gamma^\star(L)=\bar\gamma G^\star(L)$; exhaustive search uses $M\binom NL2^L$ candidates. For threshold $\gamma_{\rm th}>0$,
\begin{equation}
\begin{aligned}
 P_{\rm out}(L)
 &=\Pr\!\left\{\gamma^\star(L)<\gamma_{\rm th}\right\}\\
 &=\Pr\!\left\{G^\star(L)<t\right\},\qquad
 t=\gamma_{\rm th}/\bar\gamma.
\end{aligned}
\label{eq:outage_definition}
\end{equation}
}
{\color{black}
Under outdated CSI, the selection variables are computed from
 $\widehat{\bm g}=\sqrt{\xi_g}\bm g+\sqrt{1-\xi_g}\bm e_g$,and $\widehat{\bm H}=\sqrt{\xi_h}\bm H+\sqrt{1-\xi_h}\bm E_h$. The innovations are zero-mean complex Gaussian, independent of each other and the channels. Parameters $\xi_g,\xi_h\in[0,1]$ are squared temporal correlations. Controls maximize estimated gain but realized SNR uses the actual channels; the direct link is blocked or known.
}

{\color{black}
\section{Exact Angular Characterization}
For angle $\psi$, let $q_{m,n}(\psi)=|\Re\{e^{-j\psi}c_{m,n}\}|$ and order them as $q_{m,(1)}\ge\cdots\ge q_{m,(N)}$.
\begin{theorem}[Exact joint selection]
The fixed-port optimum is
\begin{equation}
 G_m^\star(L)=\left[\max_{0\le\psi<2\pi}\left(\Re\{e^{-j\psi}h_{d,m}\}
 +\sum_{\ell=1}^{L}q_{m,(\ell)}(\psi)\right)\right]^2.
\label{eq:angular}
\end{equation}
For an optimizing angle, activate the top-$L$ projections with $u_{m,n}=\sgn(\Re\{e^{-j\psi}c_{m,n}\})$ and set all other entries to zero, where $\sgn(x)=+1$ for $x\ge0$ and $-1$ otherwise. Finally, $G^\star(L)=\max_{m\in\mathcal M}G_m^\star(L)$.
\end{theorem}
\begin{proof}
Use $|z|=\max_\psi\Re\{e^{-j\psi}z\}$ and interchange the finite maximization over $\bm u$ with the angular maximum. For fixed $\psi$, the objective is linear in $\bm u$, so the cardinality constraint selects the $L$ largest projection magnitudes and their signs.
\end{proof}
}

{\color{black}
For any nonzero complex number $z$, let $\mathcal Z(z)=\{\arg(z)+\pi/2+k\pi:k=0,1\}\pmod{2\pi}$. Signs or ordering change only at
\begin{equation}
\begin{aligned}
 \mathcal B_m
 &=\bigcup_{n\in\mathcal N}\mathcal Z(c_{m,n})\\
 &\quad\cup\!\bigcup_{\substack{i,j\in\mathcal N\\i<j}}
 \mathcal Z(c_{m,i}-c_{m,j})\\
 &\quad\cup\!\bigcup_{\substack{i,j\in\mathcal N\\i<j}}
 \mathcal Z(c_{m,i}+c_{m,j}).
\end{aligned}
\label{eq:events}
\end{equation}
Here $\mathcal Z(0)$ is empty: an identically zero coefficient or pairwise difference creates no isolated event. Sign events are zero projections of $c_{m,n}$. Ordering events equate absolute projections; factoring their squared difference gives zero projections of $c_{m,i}-c_{m,j}$ or $c_{m,i}+c_{m,j}$, explaining \eqref{eq:events}. Each nonzero complex argument has two zero-projection angles. Thus there are at most $2N+4\binom N2=2N^2$ events. Duplicates are merged and persistent ties use a deterministic index rule. The direct term affects the evaluated sum but does not affect subset ordering. The reference sweep selects the top projections per interval with a conservative full-sort bound of $\mathcal O(MN^3\log N)$; an incremental-order implementation can reduce this to $\mathcal O(MN^2\log N)$. For $L=N$, ordering events are unnecessary: a sorted sign-event sweep costs $\mathcal O(MN\log N)$, whereas simplified DaS attains linear search \cite{xiong2024das}.
}

{\color{black}
The sweep sorts events, evaluates the original gain of each midpoint-induced top-$L$ signed vector, and retains the best port/vector.
For an optimum with nonzero sum, choose its phase $\psi^\star$. Its projection equals the optimal magnitude. Within either adjacent interval, the induced control is constant; continuity makes its projection at $\psi^\star$ equal to that magnitude, even at a boundary tie. Its complex magnitude is therefore optimal. If the optimal magnitude is zero, every feasible control is optimal. Thus interval evaluation is exact even with degeneracies; if the event set is empty, any angle suffices.
}

\begin{theorem}[Deterministic guarantees]
For $h_{d,m}=0$, let $r_{m,(1)}\ge\cdots\ge r_{m,(N)}$ order $|c_{m,n}|$. Then
\begin{equation}
 G_m^\star(L)\ge\max\!\left\{\sum_{\ell=1}^{L}r_{m,(\ell)}^2,
 \frac{4}{\pi^2}\left(\sum_{\ell=1}^{L}r_{m,(\ell)}\right)^2\right\}.
\label{eq:guarantee}
\end{equation}
\end{theorem}
\begin{proof}
For the $L$ strongest magnitudes, averaging over all signs cancels cross terms, so one sign vector attains at least their summed powers. The angular mean of $|\cos(\cdot)|$ is $2/\pi$, which gives the second bound.
\end{proof}

{\color{black}
\begin{proposition}[Reliability monotonicity]
For every channel realization and $1\le L<N$, $G^\star(L+1)\ge G^\star(L)$; hence, $P_{\rm out}(L+1)\le P_{\rm out}(L)$.
\end{proposition}
\begin{proof}
Let $z$ be the optimal $L$-element sum at its selected port and add any inactive cascaded coefficient $c_q$, $q\in\mathcal N$. Then
$\max_{b_q\in\{\pm1\}}|z+b_qc_q|^2=|z|^2+|c_q|^2+2|\Re\{z^*c_q\}|\ge|z|^2$.
\end{proof}
}

{\color{black}
\section{Outage, Diversity, and Energy Efficiency}
Assume $h_{d,m}=0$. The following outage formulas use $\bm R_R=\bm I_N$; Theorem 3 extends the diversity exponent to positive-definite $\bm R_R$. Define $X_n=|g_n|^2$, $V_n=\max_m|H_{m,n}|^2$, and their common cumulative distribution function (CDF) $F_V$; $V$ denotes a generic copy of $V_n$.
}

\begin{lemma}
For every $L\ge1$, $G^\star(L)\ge\max_nX_nV_n$.
\end{lemma}
\begin{proof}
Choose a feasible subset containing the strongest element--port path. Averaging its sum power over binary signs gives the sum of selected path powers, so at least one sign vector is no weaker than that path.
\end{proof}

If $X_n$ is exponential with mean $\Omega_g$ and $F_V$ is the CDF of $V_n$, then
\begin{equation}
 P_{\rm out}(L)\le\left[\int_0^\infty F_V\!\left(\frac{t}{x}\right)
 \frac{e^{-x/\Omega_g}}{\Omega_g}\,dx\right]^N.
\label{eq:generalbound}
\end{equation}

\begin{corollary}[Exact single-element outage]
For $L=1$, $G^\star(1)=\max_nX_nV_n$ and \eqref{eq:generalbound} holds with equality.
\end{corollary}

{\color{black}
\begin{corollary}[Independent ports]
For independent Rayleigh ports, let $a=t/(\Omega_g\Omega_h)$. Then
\begin{align}
 P_{\rm out}(L)&\le[\Psi_M(t)]^N,\label{eq:closed}\\[-1mm]
 \Psi_M(t)&=1+2\sum_{k=1}^{M}(-1)^k\binom{M}{k}\sqrt{ka}\,K_1\!\left(2\sqrt{ka}\right).\nonumber
\end{align}
Here $K_1$ is the order-one modified Bessel function of the second kind; equality holds for $L=1$.
\end{corollary}
\begin{proof}
Use $F_V(v)=(1-e^{-v/\Omega_h})^M$, expand binomially in \eqref{eq:generalbound}, and apply $\int_0^\infty e^{-\beta x-\delta/x}dx=2\sqrt{\delta/\beta}K_1(2\sqrt{\beta\delta})$.
\end{proof}
}
\begin{theorem}[Correlated-hop diversity and port asymptotics]
For fixed $\bm R_R\succ0$, $r=\operatorname{rank}(\bm R_F)\ge1$, and every fixed $1\le L\le N$, the diversity $d_L=-\lim_{\bar\gamma\to\infty}\log P_{\rm out}(L)/\log\bar\gamma=N$. The following exact $L=1$ coefficients additionally require $\bm R_R=\bm I_N$. For $r>1$,
\begin{equation}
 P_{\rm out}(1)=\left[\frac{t}{\Omega_g}\mathbb E[V^{-1}]+o(t)\right]^N.
 \label{eq:asymcorr}
\end{equation}
For $r=1$, unit diagonal implies identical port powers, and with $a=t/(\Omega_g\Omega_h)$,
\begin{equation}
 P_{\rm out}(1)=\left[a\log(1/a)+O(a)\right]^N.
 \label{eq:asymrank1}
\end{equation}
For $L>1$, these expressions are asymptotics of the outage upper bound, not exact coding-gain formulas.
\end{theorem}
\begin{proof}
Write a branch vector as $\bm h=\sqrt{\Omega_h}\bm A\bm z$, where $\bm z\sim\mathcal{CN}(\bm0,\bm I_r)$ and $\bm A\in\mathbb C^{M\times r}$ has full column rank and $\bm A\bm A^H=\bm R_F$. Norm equivalence gives positive constants $b_1,b_2$ such that $b_1\|\bm z\|^2\le V\le b_2\|\bm z\|^2$. Hence $F_V(v)=\Theta(v^r)$ near zero and $\mathbb E[V^{-1}]<\infty$ for $r>1$, following the covariance-rank reasoning of \cite{zhao2025insights}. Conditioning on $V$ gives the single-branch term $\mathbb E[1-e^{-t/(\Omega_g V)}]$. Dominated convergence after division by $t$ proves \eqref{eq:asymcorr}. If $r=1$, unit diagonal makes $V$ exponential with mean $\Omega_h$; the expansion of $1-2\sqrt a K_1(2\sqrt a)$ proves \eqref{eq:asymrank1}. Put $S=\max_n X_nV_n$. The deterministic sandwich $S\le G^\star(L)\le L^2S$ and the independent-branch CDF give exponent $N$ for every $L$.
For correlated RIS elements, factor $\bm R_F=\bm A\bm A^H$ with $\bm A\in\mathbb C^{M\times r}$ full column rank, and write $\bm H=\sqrt{\Omega_h}\bm A\bm Z$, with $\bm Z\in\mathbb C^{r\times N}$. The normalized vector $\bm x=[(\bm g/\sqrt{\Omega_g})^T,(\operatorname{vec}\bm Z)^T]^T\in\mathbb C^D$ has positive-definite covariance $\bm\Sigma=\operatorname{diag}(\bm R_R,\bm R_R^T\otimes\bm I_r)$. If $\lambda_-,\lambda_+$ are its smallest/largest eigenvalues and $D=N(r+1)$, Centered complex Gaussian densities $f_{\bm\Sigma}$ satisfy $k_-f_{\lambda_-\bm I}\le f_{\bm\Sigma}\le k_+f_{\lambda_+\bm I}$, where $k_\pm=\lambda_\pm^D/\det\bm\Sigma$. Scaling both hops by $\sqrt\lambda$ scales gain by $\lambda^2$. Integration over outage set therefore gives $k_-P_I(t/\lambda_-^2)\le P_{\rm out}(L)\le k_+P_I(t/\lambda_+^2)$, where $P_I(s)=\Pr\{G_I^\star(L)<s\}$ retains $\bm R_F$ but uses independent RIS branches. Fixed constants and threshold rescaling preserve its exponent $N$.
\end{proof}
Fixed full-rank RIS correlation changes finite-SNR reliability, not diversity. Singular RIS covariance is excluded; its rank alone does not determine the exponent. Rank-one ports add logarithmic behavior. Rank truncation changes the model; Corollary 2 is an ideal independent-port reference.
{\color{black}
For a target spectral efficiency $R_0$ in bit/s/Hz, the required SNR is $\gamma_{\rm th}=2^{R_0}-1$. The fixed-$L$ throughput, power and EE are
\begin{align}
 T(L)&=R_0[1-P_{\rm out}(L)],\qquad
 P_{\rm tot}(L)=A+LP_{\rm e},\nonumber\\[-1mm]
 \mathrm{EE}(L)&=\frac{T(L)}{P_{\rm tot}(L)},\qquad
 A=\frac{P}{\eta_{\rm PA}}+P_{\rm fix}.
\label{eq:ee}
\end{align}
Here $\eta_{\rm PA}$ is power-amplifier efficiency; $P_{\rm fix}$ includes common transceiver, full-CSI acquisition and standby costs; $P_e$ is incremental active-element power. EE is measured in bit/J/Hz.
}

{\color{black}
\subsection{QoS-Aware Adaptive Activation}
Adaptive results require perfect instantaneous CSI: outdated-CSI selection need not be feasible on the actual channel. Define
\begin{equation}
 L_{\rm Q}^{\star}=
 \min\!\left\{L\in\{1,\ldots,N\}:G^\star(L)\ge t\right\},
\label{eq:qosL}
\end{equation}
If infeasible, use $L_{\rm Q}^{\star}=N$ for power accounting and declare outage.
}

{\color{black}
\begin{proposition}[Exact adaptive search]
Bisection over $L$ finds \eqref{eq:qosL} exactly. Using the conservative full-sort bound, its complexity is $\mathcal O(MN^3\log^2N)$; an incremental-order sweep reduces it to $\mathcal O(MN^2\log^2N)$.
\end{proposition}
\begin{proof}
Reliability monotonicity makes $G^\star(L)\ge t$ a monotone feasibility test, requiring at most $\lceil\log_2N\rceil+1$ feasibility calls; the supplied implementation additionally recomputes the final control, giving at most $\lceil\log_2N\rceil+2$ sweep calls.
\end{proof}
}

{\color{black}
\begin{corollary}[Reliability, average activation, and EE]
Under perfect CSI, with $\overline L_{\rm Q}=\mathbb E[L_{\rm Q}^{\star}]$,
\begin{align}
 P_{\rm out}^{\rm Q}&=P_{\rm out}(N),\qquad
 \overline L_{\rm Q}=1+\sum_{\ell=1}^{N-1}P_{\rm out}(\ell),\label{eq:Lbar}\\
 \mathrm{EE}_{\rm Q}&=\frac{R_0[1-P_{\rm out}(N)]}{A+P_{\rm e}\overline L_{\rm Q}}\ge \mathrm{EE}(N).\label{eq:qosee}
\end{align}
The inequality is strict when $P_{\rm e}>0$ and $\Pr\{L_{\rm Q}^\star<N\}>0$. For fixed $\bm R_R\succ0$ under Theorem 3, its diversity is $N$.
\end{corollary}
\begin{proof}
Success is possible exactly when $G^\star(N)\ge t$. For $1\le\ell<N$, monotonicity gives $\Pr\{L_{\rm Q}^\star>\ell\}=P_{\rm out}(\ell)$; the tail-sum identity yields \eqref{eq:Lbar}. Finally, $\overline L_{\rm Q}\le N$ while adaptive and full-activation throughputs are equal.
\end{proof}
}

{\color{black}
\subsection{Overhead-aware interpretation}
Let $T_c$ be the coherence duration. For adaptive policy Q and full activation F, indexed by $j\in\{\mathrm Q,\mathrm F\}$, set $\tau_j=\tau_{\rm CSI}+t_{\rm solve,j}+B_j/R_{\rm ctrl}+\tau_{\rm sw,j}$ and $E_j=P_{\rm proc}t_{\rm solve,j}+E_{\rm CSI,j}^{\rm extra}+P_{\rm ctrl}B_j/R_{\rm ctrl}+E_{\rm sw,j}$. Here $\tau_{\rm CSI},t_{\rm solve,j},B_j,R_{\rm ctrl},\tau_{\rm sw,j}$ denote CSI time, solve time, control bits, control rate and switching time; $P_{\rm proc},P_{\rm ctrl}$ denote incremental processing/control powers; $E_{\rm CSI,j}^{\rm extra},E_{\rm sw,j}$ denote extra CSI and switching energies. These energies exclude costs in $A$. Hardware consumption is technology-dependent \cite{wang2024power}; $P_e$ models a specified incremental active-state cost, not universal binary-phase hardware. With $\alpha_j=1-\tau_j/T_c>0$, define
\begin{equation}
\begin{aligned}
 \mathrm{EE}^{\rm net}_j&=\alpha_jR_0[1-P_{\rm out}(N)]/D_j,\\
 D_{\rm Q}&=A+P_e\overline L_{\rm Q}+E_{\rm Q}/T_c,\\
 D_{\rm F}&=A+NP_e+E_{\rm F}/T_c.
\end{aligned}
 \label{eq:netee}
\end{equation}
Adaptation improves net EE exactly when $D_{\rm Q}<(\alpha_{\rm Q}/\alpha_{\rm F})D_{\rm F}$. For equal durations this reduces to
\begin{equation}
 \frac{E_{\rm Q}-E_{\rm F}}{T_c}<P_e(N-\overline L_{\rm Q}).\label{eq:break}
\end{equation}
Both policies require full cascaded CSI. A mask-plus-phase message uses $B_Q=2N+\lceil\log_2M\rceil$ bits versus $B_F=N+\lceil\log_2M\rceil$, before common framing. For $\tau_j<T_c$ almost surely and channel-dependent overhead, replace the numerator by $R_0\mathbb E[\alpha_j\mathbf1_{\{G^\star(N)\ge t\}}]$ and use $D_j=A+P_e\mathbb E[L_j]+\mathbb E[E_j]/T_c$. Here $L_{\rm Q}=L_{\rm Q}^\star$, $L_{\rm F}=N$, and $\mathbf1$ is the event indicator. %This retains latency--success dependence; mean-runtime plots use fixed budgets.
}

{\color{black}
\section{Numerical Results}
{\color{black}
Unless otherwise stated, the numerical results assume
a blocked direct link and $\Omega_g=\Omega_h=1$.
}
{\color{black}
\begin{figure}[t]
 \centering
 \includegraphics[width=.70\columnwidth]{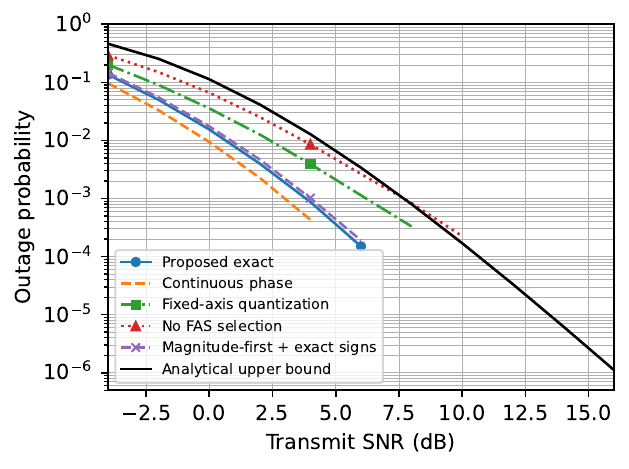}
\caption{Outage comparison for $N=4$, $M=2$, $L=2$, and
$\gamma_{\rm th}=1$. }
 \label{fig:outage}
\end{figure}
At $\bar\gamma=4$ dB in Fig.~\ref{fig:outage}, the proposed
outage is approximately $9\times10^{-4}$, compared with
$4\times10^{-3}$ for fixed-axis quantization,
$9\times10^{-3}$ without FAS selection, and
$1\times10^{-3}$ for magnitude-first activation with
exact binary phases. These comparisons demonstrate the
benefits of joint port, subset, and phase optimization:
fixing the port limits spatial selection, while
magnitude-based activation can overlook coefficients
that combine more coherently. Continuous-phase control
provides an ideal phase-resolution benchmark.

%The fixed-axis rule uses only one projection axis, whereas SA \cite{stylianopoulos2024sign} compares real- and imaginary-axis configurations. Exact binary-phase optimization on a fixed subset attains the DaS-equivalent optimum.
}

{\color{black}
\begin{figure}[t]
 \centering
 \includegraphics[width=.70\columnwidth]{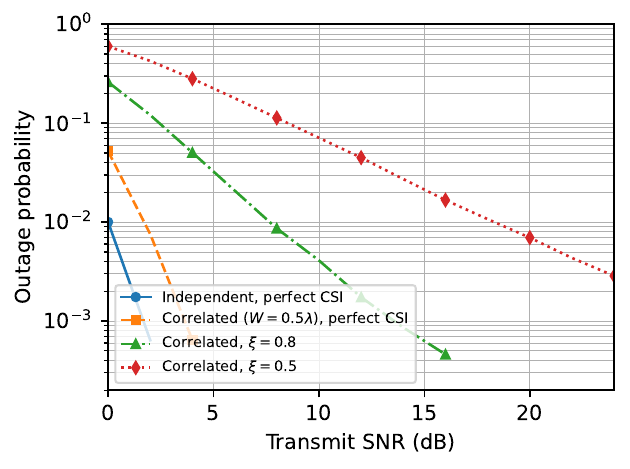}
 \caption{Fixed-$L$ effects of FAS correlation and outdated CSI for $N=M=8$, $L=4$ and $\gamma_{\rm th}=10$.}
 \label{fig:corr}
\end{figure}
At $\bar\gamma=2$ dB in Fig.~\ref{fig:corr}, the $W=0.5$
aperture raises outage relative to the independent-port
reference from approximately $6\times10^{-4}$ to
$8\times10^{-3}$. Correlated ports offer fewer distinct
favorable fades, worsening finite-SNR reliability without
changing the diversity order under the theorem's
fixed-covariance assumptions. Outdated-CSI selection
further raises outage to approximately $0.12$ for
$\xi_g=\xi_h=0.8$ and $0.43$ for $\xi_g=\xi_h=0.5$,
because the configuration selected using estimated
channels may no longer be favorable on the actual channels.
}

{\color{black}
\begin{figure*}[t]
 \centering
 \includegraphics[width=0.78\textwidth]{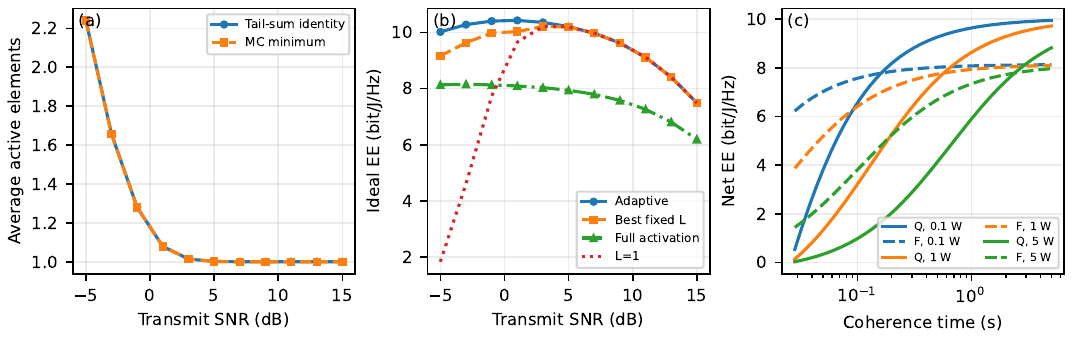}
 \caption{Perfect-CSI QoS adaptation for $N=8$, $M=4$, $W=1$, $R_0=2$ bit/s/Hz, $\gamma_{\rm th}=2^{R_0}-1=3$, $\sigma^2=1$ mW, $P=\bar\gamma\sigma^2$, $\eta_{\rm PA}=0.4$, $P_{\rm fix}=0.18$ W, $P_{\rm e}=8$ mW, and $3\times10^4$ exact realizations: (a) average activation, (b) ideal EE, and (c) net EE.}
 \label{fig:qos}
\end{figure*}
}

{\color{black}
Fig.~\ref{fig:qos}(a) verifies \eqref{eq:Lbar}.
At $-5$ dB, full activation has an outage probability
of approximately $5\times10^{-3}$, whereas adaptation
preserves that reliability using about $2.2$ elements
on average. In Fig.~\ref{fig:qos}(b), the adaptive policy achieves approximately
$10$ bit/J/Hz, compared with $9.2$ for the best fixed-$L$
policy and $8.1$ for full activation.
Fig.~\ref{fig:qos}(c) accounts for serial solver latency
and energy for both policies. Processor powers of
$0.1$, $1$, and $5$ W and input/output costs are assumed. Short coherence
blocks can erase the activation saving, whereas longer
blocks amortize processing overhead. At high SNR,
$\overline L_{\rm Q}\to1$, while increasingly demanding
SNR targets require more active elements and eventually
lead to full activation.
}

{\color{black}
\begin{figure}[t]
\centering\includegraphics[width=.70\columnwidth]{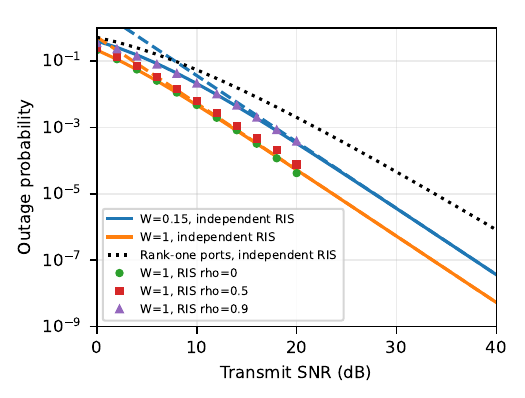}
\caption{Correlated-port single-element outage: $N=2$, $M=4$, $L=1$, $\gamma_{\rm th}=1$.}\label{fig:asym}
\end{figure}
In Fig.~\ref{fig:asym}, increasing RIS correlation
$\rho$ on both hops worsens finite-SNR outage in the
considered cases, but the diversity order remains $N$
for fixed positive-definite RIS covariance.
This order is established analytically by Theorem~3,
rather than inferred from finite-SNR curve slopes.
\begin{figure}[t]
\centering\includegraphics[width=.70\columnwidth]{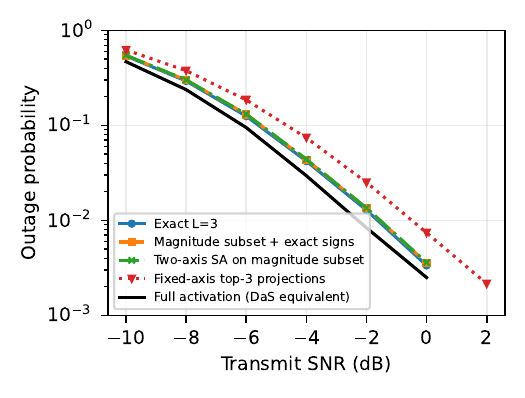}
\caption{Correlated-channel benchmarks: $N=M=4$, $W=0.5$, $\gamma_{\rm th}=1$.}\label{fig:benchcorr}
\end{figure}

In Fig.~\ref{fig:benchcorr}, exact selection achieves
a gain no smaller than any feasible heuristic on each
channel realization. Magnitude-first activation can
miss subsets whose coefficients combine more coherently.
The two-axis SA rule \cite{stylianopoulos2024sign} is
applied to the magnitude-selected subset, whereas the
fixed-axis top-$L$ rule selects elements by their absolute
real projections. Full activation uses one additional
element and provides a reliability reference with a
different activation-power cost. The point-to-point
FAS formulas in \cite{zhao2025insights} describe a
noncascaded link and are compared analytically rather
than used as numerical benchmarks for the cascaded system.
}

{\color{black}
\section{Conclusion}
This letter developed an exact angular event sweep for
joint FA port selection, fixed-cardinality
RIS activation, and 1-bit phase control. The analysis
established deterministic gain bounds, activation
monotonicity, and diversity order for fixed
positive-definite RIS covariance and fixed FAS covariance
of any nonzero rank. Under perfect CSI,
minimum-element QoS adaptation preserves the outage
probability of full activation while reducing average
activation. Net EE improvement depends
on whether these savings outweigh the additional
processing, signaling, and reconfiguration overhead.
}
}
\bibliographystyle{IEEEtran}
\bibliography{references}
\end{document}